\documentclass[conference]{IEEEtran}
\ifCLASSINFOpdf
\else
\fi
\usepackage{graphicx} % support the \includegraphics command and options
\usepackage{color}

\definecolor{Red}{rgb}{1,0,0}
\definecolor{Blue}{rgb}{0,0,1}
\definecolor{Olive}{rgb}{0.41,0.55,0.13}
\definecolor{Green}{rgb}{0,1,0}
\definecolor{MGreen}{rgb}{0,0.8,0}
\definecolor{DGreen}{rgb}{0,0.55,0}
\definecolor{Yellow}{rgb}{1,1,0}
\definecolor{Cyan}{rgb}{0,1,1}
\definecolor{Magenta}{rgb}{1,0,1}
\definecolor{Orange}{rgb}{1,.5,0}
\definecolor{Violet}{rgb}{.5,0,.5}
\definecolor{Purple}{rgb}{.75,0,.25}
\definecolor{Brown}{rgb}{.75,.5,.25}
\definecolor{Grey}{rgb}{.5,.5,.5}

\usepackage{booktabs} % for much better looking tables
\usepackage{array} % for better arrays (eg matrices) in maths
\usepackage{paralist} % very flexible & customisable lists (eg. enumerate/itemize, etc.)
\usepackage{verbatim} % adds environment for commenting out blocks of text & for better verbatim
\usepackage{subfigure} % make it possible to include more than one captioned figure/table in a single float
\usepackage[cmex10]{amsmath}
\usepackage{amsthm}

\theoremstyle{plain}

\newtheorem{theorem}{Theorem}

\newtheorem{corollary}{Corollary}
\newtheorem{claim}{Claim}
\newtheorem{lemma}{Lemma}
\newtheorem{conj}{Conjecture}
\newtheorem{bound}{Bound}

\theoremstyle{remark}

\newtheorem{remark}{Remark}

\theoremstyle{definition}
\newtheorem{definition}{Definition}

\usepackage{amssymb}
\usepackage{xspace}
\usepackage{mathrsfs}
\usepackage{amsopn}

\newcommand{\Mc}{\mathcal{M}}

\newcommand{\Uc}{\mathcal{U}}
\newcommand{\Vc}{\mathcal{V}}
\newcommand{\Wc}{\mathcal{W}}
\newcommand{\Xc}{\mathcal{X}}
\newcommand{\Yc}{\mathcal{Y}}
\newcommand{\Zc}{\mathcal{Z}}

\def\a{\alpha}

\def\eps{\epsilon}
\def\la{\lambda}

\def\ce{\mathfrak{C}}

\DeclareMathOperator\E{E}
\let\P\relax
\DeclareMathOperator\P{P}
\def\textiid{i.i.d.\@\xspace}
\newcommand\iid{\ifmmode\text{ i.i.d. } \else \textiid \fi}

\newcommand{\qmf}{\mathfrak{q}}

\everymath{\textstyle}

\begin{document}
\title{On Marton's inner bound for broadcast channels}

\author{
\IEEEauthorblockN{Amin Gohari}
\IEEEauthorblockA{Dept. of EE\\
Sharif University of Technology\\
Tehran, Iran}
\and
\IEEEauthorblockN{Chandra Nair}
\IEEEauthorblockA{Dept. of IE\\
The Chinese University of Hong Kong\\
Hong Kong}
\and
\IEEEauthorblockN{Venkat Anantharam}
\IEEEauthorblockA{Dept. of EECS\\
University of California, Berkeley\\
Berkeley, CA, USA}}


%Telephone: (800) 555--1212, Fax: (888) 555--1212}


\maketitle
% !TEX TS-program = pdflatex % !TEX encoding = UTF-8 Unicode

% This is a simple template for a LaTeX document using the "article" class.
% See "book", "report", "letter" for other types of document.

\begin{abstract}
Marton's inner bound is the best known achievable region for a general
discrete memoryless broadcast channel. To compute Marton's inner bound
one has to solve an optimization problem over a set of joint distributions on
the input and auxiliary random variables. The optimizers turn out to
be structured in many cases. Finding properties of optimizers not only
results in efficient evaluation of the region, but it may also help one to
prove factorization of Marton's inner bound (and thus its optimality).
The first part of this paper formulates this factorization approach
explicitly and states some conjectures and results along this line.
The second part of this paper focuses primarily on the structure of
the optimizers. This section is inspired by a new binary inequality
that recently resulted in a very simple characterization of the sum-rate
of Marton's inner bound for binary input broadcast channels.
This prompted us to investigate whether this inequality can be extended to
larger cardinality input alphabets. We show that several of the results for the
binary input case do carry over for higher cardinality alphabets and we
present a collection of results that help restrict the search space of
probability distributions to evaluate the boundary of Marton's
inner bound in the general case. We also prove a new inequality
for the binary skew-symmetric broadcast channel that yields a very
simple characterization of the entire Marton inner bound for this channel.
\end{abstract}
\IEEEpeerreviewmaketitle

\section{Introduction}
A broadcast channel \cite{cov72} models a communication scenario where a single sender wishes to communicate multiple messages to many receivers. A two receiver discrete memoryless broadcast channel consists of a sender $X$ and two receivers $Y, Z$. The sender maps a pair of messages $M_1, M_2$ to a transmit sequence $X^n(m_1,m_2) (\in \Xc^n)$ and the receivers each get a noisy version $Y^n (\in \Yc^n), Z^n (\in \Zc^n)$ respectively. Further $|\Xc|, |\Yc|, |\Zc| < \infty$ and $p(y_1^n,z^n|x^n) = \prod_{i=1}^n p(y_{i},z_{i}|x_i)$. For more details on this model and a collection of known results please refer to Chapters 5 and 8 in \cite{elk12}. We also adopt most of our notation from this book.

The best known achievable rate region for a broadcast channel is the following inner bound due to \cite{mar79}. Here we consider the private messages case.
\begin{bound}
\label{bd:MIB}
(Marton) The union of rate pairs $R_1, R_2$ satisfying the constraints
\begin{align*}
R_1 & < I(U,W;Y), \\
R_2 & < I(V,W;Z), \\
R_1 + R_2 & < \min\{I(W;Y), I(W;Z)\} + I(U;Y|W) \\
& \quad + I(V;Z|W) - I(U;V|W),
\end{align*}
for any triple of random variables $(U,V,W)$ such that $(U,V,W) \to X \to (Y, Z)$ is achievable. Further to compute this region it suffices \cite{goa09} to consider $|\Wc| \leq |\Xc|+4, |\Uc| \leq |\Xc|, |\Vc| \leq |\Xc|$.
\end{bound}

%It is not known whether this region is the true capacity region or not since the traditional Gallager type technique for proving converses fail here.

It is not known whether this region is the true capacity region since the traditional Gallager-type technique for proving converses fails to work in this case. This raises the question of whether Marton's inner bound has an alternative representation that is better amenable to analysis. \iffalse Could it even be the case that Marton's inner bound matches the true capacity region? If so, is it possible that we just have not got the right way to go about proving a converse to Marton's inner bound.\fi We believe that central to answering this question is understanding properties of joint distributions $p(u,v,w,x)$ corresponding to extreme points of Marton's inner bound.
Our approach to this is twofold. Roughly speaking in the first part of this paper we find \emph{sufficient} conditions on the optimizing distributions $p(u,v,w,x)$ which would imply a kind of factorization of Marton's inner bound.
 Such a factorization would imply that Marton's region is the correct rate region. In the second part we find \emph{necessary} conditions on any optimizing $p(u,v,w,x)$. Unfortunately the gap between these sufficient and necessary conditions is still wide. However we discuss how the necessary conditions may enhance our understanding of the maximizers of the expression $I(U;Y)+I(V;Z)-I(U;V)$ and how it may be useful in proving the factorization of Marton's inner bound.
%
%\begin{remark}
%Due to page limitations we have to omit some of the proofs from this version. They can be found in the full version uploaded on arXiv.
%\end{remark}

\subsection{Necessary conditions}
%Recently\cite{nae07,naw08,goa09,jon09,nwg10} there has been an interest in evaluating Marton's inner bound and comparing the region to known outer bounds for this channel.
The question of whether Marton's inner bound matches one of the known outer bounds has been studied in several works recently \cite{nae07,naw08,goa09,jon09,nwg10}. Since we build upon these results in this work, a brief literature review is in order. It was shown in \cite{naw08} that a gap exists between Marton's inner bound and the best-known outer bound \cite{elg79} for the binary skew-symmetric (BSSC) broadcast channel (Fig. \ref{fig:bsscbc}) if a certain binary inequality, \eqref{eq:eq1} below, holds. A gap between the bounds was demonstrated for the BSSC in \cite{goa09} without explicitly having to evaluate the inner bound. The conjectured inequality for this channel was established in \cite{jon09} and hence Marton's sum-rate for BSSC was explicitly evaluated. The inequality was shown \cite{nwg10} to hold for all binary input broadcast channels thus giving an alternate representation to Marton's sum-rate for binary input broadcast channels.

\begin{theorem}\cite{nwg10}
\label{th:beq}
For all random variables $(U,V,X,Y,Z)$ such that $(U,V) \to X \to (Y,Z)$ and $|\Xc|=2$ the following holds
{\small \begin{equation}
\label{eq:eq1}
I(U;Y) + I(V;Z) - I(U;V) \leq \max \{ I(X;Y), I(X;Z) \}.
\end{equation}}
\end{theorem}

This yields the following immediate corollary.
\begin{corollary}\cite{nwg10}
\label{co:binsum}
The maximum sum-rate achievable by Marton's inner bound for any binary input broadcast channel is given by
\begin{align*} \max_{p(w,x)} &\min\{I(W;Y), I(W;Z)\} + \P(W=0) I(X;Y|W=0) \\
& \qquad + \P(W=1) I(X;Z|W=1). \end{align*}
Here $\Wc =\{0,1\}$.
\end{corollary}
Note that this characterization is  much simpler than the one given in Bound \ref{bd:MIB}.

Our results on the necessary conditions of an optimizer attempt to extend the new binary inequality to larger alphabets and to the entire rate region (rather than just the sum rate).

\subsection{Sufficient conditions}
Suppose we have certain properties of $p(u,v,w,x)$ that maximize Marton's inner bound. How can one use this to prove that Marton's inner bound is tight? The traditional Gallager-type technique requires us to consider the $n$-letter expression and to try to identify single-letter auxiliary random variables. If any such statement can be shown, it has to hold for $n=2$ in particular. \iffalse But the traditional Gallager type technique imposes certain structures on the choice of single-letter auxiliary random variables (they have to consist of past/future of the auxiliary random variables). To relax this constraint the authors\fi In \cite{ggny11a}, the authors studied Marton's inner bound (sum-rate) via a two-letter approach and there they presented an approach to test whether Marton's inner bound is indeed optimal. The crux of the paper \cite{ggny11a}  is a certain {\em factorization idea} which if established would yield the optimality of Marton's inner bound for discrete memoryless broadcast channels. Further the authors used the same idea to show \cite{ggny11b} an example of a class of broadcast channels where Marton's inner bound is tight and the best known outer bounds are strictly loose\footnote{The previous works established a gap between the bounds and in this work it was shown that the outer bounds (both in the presence and absence of a common message) are strictly sub-optimal.}. The converse to the capacity region of this class of broadcast channels was motivated by the factorization approach. The authors also showed that the factorizing approach works if an optimizer $p(u,v,w,x_1x_2)$ for the two-letter Marton's inner bound satisfies certain conditions.

In this paper we provide more sufficient conditions that imply factorization by forming a more refined version of the two-letter approach \cite{ggny11b}. Simulations conducted on randomly generated binary input broadcast channels indicate that perhaps the factorization stated below (Conjecture \ref{co:conj}) is true; thus indicating that Marton's inner bound could be optimal.

For any broadcast channel $\qmf(y,z|x)$, define
$$T(X) := \max_{p(u,v|x)} I(U;Y) + I(V;Z) - I(U;V). $$
Note that $T(X)$ is a function of $p(x)$ for a given broadcast channel.
Similarly for any function $f(X)$, defined on $p(x)$ denote by
$$ \ce[f(X)] := \max_{p(v|x)} \sum_{v} p(v) f(X|V=v), $$
the upper concave envelope evaluated of $f(X)$ at $p(x)$. (Note that one can restrict the maximization to $|\Vc| \leq |\Xc|$ by Fenchel-Caratheodory arguments). A 2-letter broadcast channel is a product broadcast channel whose transition probability is given by
$\qmf(y_{1},z_{1}|x_1) \qmf(y_{2},z_{2}|x_2)$; i.e. they can be considered as parallel non-interfering broadcast channels. For this channel the function $T(X_1, X_2)$ is defined similarly as
$$\max_{p(u,v|x_1, x_2)} I(U;Y_{1},Y_{2}) + I(V;Z_{1},Z_{2}) - I(U;V). $$
%\begin{align*}
%\resizebox{\hsize}{!}{$T(X_1, X_2) :=\max_{p(u,v|x_1, x_2)} I(U;Y_{1},Y_{2}) + I(V;Z_{1},Z_{2}) - I(U;V).$}
%\end{align*}
%$$T(X_1, X_2) :=\max_{p(u,v|x_1, x_2)} I(U;Y_{1},Y_{2}) + I(V;Z_{1},Z_{2}) - I(U;V). $$
\begin{conj}
\label{co:conj}
For all product channels, for all $\la \in [0,1]$ and for all $p(x_1,x_2)$
the following holds:
\begin{align*}
 &- \la H(Y_{1},Y_{2}) - \bar \la H(Z_{1}, Z_{2}) + T(X_1, X_2) \\
&  \leq \ce[ -\la H(Y_{1}) - \bar \la H(Z_{1}) + T(X_1) ] \\
& \qquad + \ce[ -\la H(Y_{2}) - \bar \la H(Z_{2}) + T(X_2)],
\end{align*}
where $\bar \la = 1-\la.$
\end{conj}

\begin{remark}
The above conjecture was not formally stated  in \cite{ggny11a} as the authors did not have enough numerical evidence at that point; however subsequently the evidence has grown enough for some of the authors to have reasonable confidence in the validity of the above statement.
\end{remark}

It was shown \cite{ggny11a} that if Conjecture \ref{co:conj} holds then Marton's inner bound would yield the optimal sum-rate for a two-receiver discrete memoryless broadcast channel. Hence establishing the veracity of the conjecture becomes an important direction in studying  the optimality of Marton's inner bound.

The validity of Conjecture \ref{co:conj} was established \cite{ggny11a} in the following three instances:
\begin{enumerate}
\item $\la=0, \la = 1$, i.e. the extreme points of the interval,
\item If one of the four channels, say $X_1 \mapsto Y_{1}$ is deterministic,
\item In one of the components, say the first, receiver $Y_{1}$ is more capable\footnote{A receiver $Y$ is said to be {\em more-capable} \cite{kom75} than receiver $Z$ if $I(X;Y) \geq I(X;Z)~ \forall p(x).$} than receiver $Z_{1}$.
\end{enumerate}

Note that to establish the conjecture one needs to get a better handle on $T(X)$. What inequality \eqref{eq:eq1} shows is that when $|X|=2$ then
$$ T(X) = \max\{I(X;Y), I(X;Z)\}. $$
In this work, we seek generalizations of the inequality \eqref{eq:eq1} in two different directions:
\begin{itemize}
  \item {\em To the entire private messages region}:  Maximizing $I(U;Y) + I(V;Z) - I(U;V)$ for a given $p(x)$ is related to the sum-rate computation of Marton's inner bound. If one is interested in the entire private messages region, one must deal with a slightly more general form and this is presented in Section \ref{sse:gen}.
\item {\em Beyond binary input alphabets}:  The inequality \eqref{eq:eq1} itself fails to hold where $|\Xc|=3$, for instance in the Blackwell channel\footnote{Blackwell channel is a deterministic broadcast channel with $\Xc=\{0,1,2\}$, with the  mapping $\Xc \mapsto \Yc \times \Zc$ given by: $0 \mapsto (0,0), 1 \mapsto(0,1), 2 \mapsto (1,1)$.}. Therefore, we attempt to establish properties of the optimizing distributions $p(u,v|x)$ that achieve $T(X)$, in Section \ref{se:hc}.
\end{itemize}

\subsubsection{A generalized conjecture}
\label{sse:gen}
Much of the work in \cite{ggny11a} focused on the sum-rate. If one is interested in proving the optimality of the entire rate-region (for the private message case) then establishing the following equivalent conjecture would be sufficient. For $\alpha \geq 1$ define $T_\alpha(X) := \max_{p(u,v|x)} \alpha I(U;Y) + I(V;Z) - I(U;V).$
\begin{conj}
\label{co:marfac}
For all product channels, for all $\la \in [0,1]$, for all $\alpha \geq 1$, and for all $p(x_1,x_2)$ the following holds:
\begin{align*}
 &-(\alpha - \bar \la) H(Y_{1},Y_{2}) - \bar \la H(Z_{1}, Z_{2}) + T_\alpha(X_1, X_2) \\
& \quad \leq \ce[ -(\alpha - \bar \la) H(Y_{1}) - \bar \la H(Z_{1}) + T_\alpha(X_1) ] \\
& \qquad+ \ce[ -(\alpha - \bar\la) H(Y_{2}) - \bar \la H(Z_{2}) + T_\alpha(X_2) ].
\end{align*}
\end{conj}

\begin{remark}
The sufficiency of the conjecture in proving the optimality of Marton's inner bound follows from a 2-letter argument similar to that found in \cite{ggny11a}.
However this conjecture is not equivalent to proving the optimality of Marton's inner bound; indeed it is a stronger statement.
\end{remark}

\section{Sufficient conditions}
A sufficient condition beyond those established in \cite{ggny11a} that imply factorization is the following:
\begin{claim}
For some $p(x_1,x_2)$ and a product channel if we have a $p(u,v|x_1,x_2)$ such
that
$$ T(X_1,X_2) = I(U;Y_1,Y_2) + I(V;Z_1, Z_2) - I(U;V),$$
and further $\P(X_2=x_2|U=u) \in \{0,1\}~\forall u,x_2$, then the factorization
conjecture holds.
\end{claim}
\begin{proof}
Observe that (by elementary manipulations) we have
\begin{align*}
& -(\alpha - \bar \la) H(Y_{1},Y_{2}) - \bar \la H(Z_{1}, Z_{2}) + \alpha I(U;Y_{1},Y_{2})\\
& \quad + I(V;Z_{1},Z_{2}) - I(U;V)\\
%& \quad = -(\alpha - \bar \la) H(Y_{1}|Z_{2}) - \bar \la H(Z_{1}| Z_{2}) + \alpha I(U;Y_{1}|Z_{2}) \\
%& \qquad + I(V;Z_{1}|Z_{2}) - I(U;V|Z_{2})  -(\alpha - \bar \la) H(Y_{2}|Y_{1})  \\
%&\qquad- \bar \la H(Z_{2}| Y_{1}) + \alpha I(U;Y_{2}|Y_{1})
%+ I(V;Z_{2}|Y_{1}) \\
%&\qquad - I(U;V|Y_{1}) + I(U;V|Y_{1}, Z_{2}) - (\alpha-1)I(Y_{1};Z_{2}|U)\\&\qquad -
%I(Y_{1};Z_{2}|U,V) \\
& \quad = -(\alpha - \bar \la) H(Y_{1}|Z_{2}) - \bar \la H(Z_{1}| Z_{2}) + \alpha I(U;Y_{1}|Z_{2})\\
&\qquad + I(V;Z_{1}|Z_{2}) - I(U;V|Z_{2}) -(\alpha - \bar \la) H(Y_{2}|Y_{1})\\
&\qquad  - \bar \la H(Z_{2}| Y_{1}) + \alpha I(U;Y_{2}|Y_{1})
+ I(V;Z_{2}|U,Y_{1})\\
&\qquad  - (\alpha-1)I(Y_{1};Z_{2}|U) -
I(Y_{1};Z_{2}|U,V).
\end{align*}
Since $X_2$ is a function of $U$ we have
\begin{align*}
&  \alpha I(U;Y_{2}|Y_{1}) + I(V;Z_{2}|U,Y_{1}) - (\alpha-1)I(Y_{1};Z_{2}|U)\\
& \quad  -
I(Y_{1};Z_{2}|U,V) = \alpha I(X_2;Y_{2}|Y_{1}).\end{align*}
Hence
\begin{align*} & T(X_1,X_2) = -(\alpha - \bar \la) H(Y_{1}|Z_{2}) - \bar \la H(Z_{1}| Z_{2})\\
& \quad  + \alpha I(U;Y_{1}|Z_{2}) + I(V;Z_{1}|Z_{2}) - I(U;V|Z_{2}) \\
&\qquad -(\alpha - \bar \la) H(Y_{2}|Y_{1})\\&\qquad - \bar \la H(Z_{2}| Y_{1}) +  \alpha I(X_2;Y_{2}|Y_{1}) \\
& \leq \ce[ -(\alpha - \bar \la) H(Y_{1}) - \bar \la H(Z_{1}) + T_\alpha(X_1) ] \\
& \quad + \ce[ -(\alpha - \bar\la) H(Y_{2}) - \bar \la H(Z_{2}) + T_\alpha(X_2) ].
\end{align*}
\end{proof}
\begin{remark}
The main purpose of this claim is to demonstrate that if the distributions $p(u,v|x)$ that achieve $T(X)$, we will refer to them as {\em extremal distributions}, satisfy certain  properties, then we could employ these properties to establish the conjecture. In this paper we will establish some such properties of the extremal distributions.
\end{remark}

\subsection{A conjecture for binary alphabets}
A natural guess for extending the inequality \eqref{eq:eq1}, so as to compute $T_\alpha(X)$, is the following:
for any $\alpha \geq 1$, for all random variables $(U,V,X,Y,Z)$ such that $(U,V) \to X \to (Y,Z)$ and $|\Xc|=2$, the following holds
\begin{align}
\alpha I(U;Y) + I(V;Z) - I(U;&V) \leq\nonumber\\& \max \{ \alpha I(X;Y),I(X;Z) \}.\label{eq:eqg}
\end{align}
However this inequality turns out to be false in general. A counterexample is presented in Appendix B.

However the inequality is true in the following cases:
\begin{enumerate}
\item If $\alpha \leq 1$ then the inequality in \eqref{eq:eqg} holds: To see this let $Y'$ be obtained from $Y$ by erasing each received symbol with probability $1-\alpha$. It is straightforward to see that $I(U;Y') = \alpha I(U;Y)$ and $I(X;Y') = \alpha I(X;Y)$. Since
$$ I(U;Y') + I(V;Z) - I(U;V) \leq \max\{I(X;Y'), I(X;Z)\},$$
the inequality holds.
\item If $\alpha \geq 1$ at any $p(x)$ where $I(X;Y) \geq I(X;Z)$ the inequality holds since
\begin{align*}
& (\alpha - 1) I(U;Y) \leq (\alpha - 1) I(X;Y),  \\
& I(U;Y) + I(V;Z) - I(U;V) \leq I(X;Y). \end{align*}
\end{enumerate}

The inequality in Equation \eqref{eq:eqg} also holds for the binary skew-symmetric broadcast channel shown in Figure \ref{fig:bsscbc} (we assume $p=\frac 12$); quite possibly the simplest channel whose capacity region is not established. The proof is presented in Appendix A.
\begin{figure}
  \begin{center}
    \includegraphics[width=0.50\linewidth]{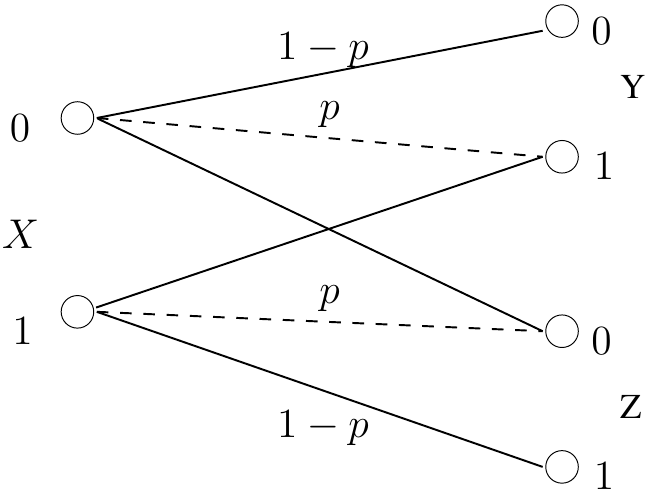}
  \end{center}
  \caption{The binary skew-symmetric broadcast channel}
\label{fig:bsscbc}
\end{figure}

By establishing  Equation \eqref{eq:eqg} for this channel, we are now able to precisely characterize Marton's inner bound region for this channel. In particular it is straightforward to see that for $\alpha \geq 1,$ if $\Mc$ represents Marton's inner bound, then

\begin{align*}  & \max_{(R_1,R_2) \in \Mc}  \alpha R_1 + R_2 \\
&=\quad  \max_{p(w,x)} ~~  \Big(\min\{I(W;Y), I(W;Z)\} + (\alpha-1) I(W;Y)  \\
& \qquad\qquad\qquad+ \alpha \P(W=0)I(X;Y|W=0) \\&\qquad\qquad\qquad+ \P(W=1)I(X;Z|W=1)\Big).
\end{align*}
A similar statement holds for when the roles of $Y,Z$ are interchanged.
%In particular this may enable one to show the sub-optimality of Marton's inner bound by looking at directions beyond %the sum-rate.

Based on simulations and other evidence we propose the following conjecture.
\begin{conj}
\label{co:exbin}
For all $\alpha \geq 1$, for all $(U,V) \to X \to (Y,Z)$ with $|\Xc|=2$, we have
\begin{align*}
&-(\alpha - \bar \la) H(Y) - \bar \la H(Z) + T_\alpha(X) \\
&  \leq \ce[-(\alpha - \bar \la) H(Y) - \bar \la H(Z) + \max\{\alpha I(X;Y), I(X;Z)\}].
\end{align*}
\end{conj}

\begin{remark}
Clearly for a broadcast channel if equation \eqref{eq:eqg} holds then the conjecture holds. Even though we know that Equation \eqref{eq:eqg} may fail at some $p(x)$ for some channels,  the conjecture states that Equation \eqref{eq:eqg} holds for a sufficient class of $p(x)$ that is needed to compute the concave envelope.
\end{remark}

\section{Necessary conditions: beyond binary input alphabets}
\label{se:hc}
In this section we compute some properties of the extremal distributions for $T(X)$, $|\Xc| \geq 3$. To understand our approach, it is useful to have a quick recap of the proof of Equation \eqref{eq:eq1} for binary alphabets. The main idea behind the proof is to isolate the local maxima of the function $p(u,v|x)$ by a perturbation argument, an extension of the ideas introduced in \cite{goa09}. The following facts were established in \cite{goa09}: for a fixed broadcast channel $\qmf(y,z|x)$ to compute
$$ \max_{p(u,v|x)} I(U;Y) + I(V;Z) - I(U;V) $$ if suffices to consider
\begin{enumerate}
  \item $|\Uc|, |\Vc| \leq |\Xc|,$ and
 \item $p(x|u,v) \in \{0,1\}$, i.e. $X$ is a function of $(U,V)$, say $X=f(U,V)$.
\end{enumerate}
When $X$ is binary, there are 16 possible functions from $U,V$ to $X$. The proof \cite{nwg10} essentially boiled down to showing that the local maxima may only exist for the following two cases:
$U = X, V = \emptyset; ~ V=X, U=\emptyset$, leading to the terms $I(X;Y), I(X;Z)$ respectively. Indeed, in the proof, there were only two non-trivial cases to eliminate: these were (assume w.l.o.g. all alphabets of $\Uc, \Vc, \Xc$ are $\{0,1\}$):
$$ X = U \oplus V ~\mbox{(XOR case)}, \quad  X = U \wedge V ~\mbox{(AND case}). $$
Hence we adopt the approach of eliminating classes of functions where the local maxima may exist and we present the generalizations of the AND case and the XOR cases in the next two sections.

\smallskip

In the following sections we assume that $p(u,v|x)$ achieves $T(X)$ and $X=f(U,V)$. Further we assume that $\qmf(y,z|x) > 0 ~\forall x,y,z$, i.e. we are in a dense subset of channels with non-zero transition probabilities. In this case we can further assume that $p(u,v) > 0 ~\forall u,v$, \cite{gea10}.

\subsection{Generalization of the AND case}
In this section we deal with an extension of the AND case from the proof of the binary inequality~\cite{nwg10}. It says that one cannot have one column and one row mapped to the same input symbol.
\begin{theorem} \label{theoremIm1} For any $(U,V,X)$ such that $X=f(U,V)$ and $p(uv|x)$ achieves $T(X)$ one cannot find $x_0$, $u_0$ and $v_0$ such that $f(u_0,v)=f(u,v_0)=x_0$ for all $u\in \mathcal{U}$ and $v\in \mathcal{V}$.
\end{theorem}

\begin{proof}
Assume otherwise that $f(u_0,v)=f(u,v_0)=x_0$ for all $u\in \mathcal{U}$ and $v\in \mathcal{V}$. Consider the multiplicative perturbation $q_{u,v,x}=p_{u,v,x} (1+\varepsilon L_{u,v})$ for some
$\varepsilon$ in some interval around zero. For this to be a valid perturbation, it has to preserve the marginal distribution of $X$. Therefore we require that
\begin{align}\sum_{u,v}p_{u,v,x}L_{u,v}=0~~~~\forall x,\label{EqI0}
\end{align}
%The proof uses a lemma that is provided in the next section at the end of the proof.

We can view the expression $I(U;Y)+I(V;Z)-I(U;V)$ evaluated at $q_{u,v,x}$ as a function of $\varepsilon$. \iffalse Since $p(u,v,x)$ was maximizing the expression $I(U;Y)+I(V;Z)-I(U;V)$, the first derivative of $I(U;Y)+I(V;Z)-I(U;V)$ with respect to $\varepsilon$ must be zero, at $\varepsilon=0$. Its second derivative must be non-positive. The second derivative can be computed using the perturbation method\cite{goa09}.\fi Non-positivity of the second derivative at a local maximum implies
\begin{align*}&\mathbb{E}(\mathbb{E}(L|U,Y)^2)+\mathbb{E}(\mathbb{E}(L|V,Z)^2)-\mathbb{E}(\mathbb{E}(L|U,V)^2)\leq 0.
\end{align*}
where random variable $L$ is defined to take the value $L_{u,v}$ under the event that $(U,V)=(u,v)$. Routine calculations show that this condition can be rewritten as follows
\begin{align}\label{eqnI:M0}
&\sum_{u,v}\frac{1}{p_{uv}}I_{u,v}^2-\sum_{u}\sum_{v_1}\sum_{v_2}T_{f(u,v_1),f(u,v_2),u}I_{u,v_1}I_{u,v_2}\\
& - \sum_{v}\sum_{u_1}\sum_{u_2}T_{f(u_1,v),f(u_2,v),v}I_{u_1,v}I_{u_2,v}\geq 0, \nonumber
\end{align}
where $I_{u,v}=p_{uv}L_{uv}$, $T_{x_1,x_2,u} = \sum_{y} p_{y|x_1}p_{y|x_2} \frac{1}{p_{uy}}$, and $T_{x_1,x_2,v}$ is defined similarly.
Equation \eqref{EqI0} can be rewritten as \begin{align}\label{EqI1}\sum_{u,v: x=f(u,v)}I_{u,v}=0~~~~\forall x.\end{align}

Now, let us define $I_{u,v}$ as follows: (a) $I_{u,v}=0$ when $u\neq u_0$ and $v\neq v_0$, (b) $I_{u_0,v}=p_{u_0,v}p_{v_0}$ when $v\neq v_0$,
(c) $I_{u,v_0}=-p_{u,v_0}p_{u_0}$ when $u\neq u_0$, and (d) $I_{u_0,v_0}=p_{u_0v_0}(p_{v_0}-p_{u_0}).$ Note that $I_{u_0,v}>0$ for all $v\neq v_0$, and $I_{u,v_0}<0$ for all $u\neq u_0$ since $p_{u,v}>0$.

It is easy to verify equation \eqref{EqI1} for this choice.
The second derivative constraint reduces (after some manipulation) to
\begin{align}\nonumber
& \sum_{\substack{u,v:~u=u_0\\or~v=v_0}}\frac{1}{p_{uv}}I_{u,v}^2 \geq\sum_{u:u\neq u_0} T_{x_0,x_0,u}I_{u,v_0}^2  \\
& \quad \qquad +\sum_{v:v\neq v_0} T_{x_0,x_0,v}I_{u_0,v}^2 +T_{x_0,x_0,u_0}(\sum_{v}I_{u_0,v})^2 \nonumber\\&\quad \qquad +T_{x_0,x_0,v_0}(\sum_{u}I_{u,v_0})^2. \label{eqnI:M1}
\end{align}

Now, using Lemma \ref{le:kld} (a very similar result was used in \cite{nwg10}) one can see that $T_{x_0,x_0,v}\geq \frac{p_{u_0v_0}}{p_{u_0v}p_{v_0}}$, $T_{x_0,x_0,v_0}\geq \frac{1}{p_{v_0}}$,  $T_{x_0,x_0,u}\geq \frac{p_{u_0v_0}}{p_{uv_0}p_{u_0}}$ and $T_{x_0,x_0,u_0}\geq \frac{1}{p_{u_0}}$. Hence, observe that
\begin{align}\nonumber&
\sum_{u:u\neq u_0} T_{x_0,x_0,u}I_{u,v_0}^2+T_{x_0,x_0,u_0}(\sum_{v}I_{u_0,v})^2\\
&\quad +\sum_{v:v\neq v_0} T_{x_0,x_0,v}I_{u_0,v}^2+T_{x_0,x_0,v_0}(\sum_{u}I_{u,v_0})^2\nonumber\\& \geq
\sum_{u:u\neq u_0} \frac{p_{u_0v_0}}{p_{uv_0}p_{u_0}}I_{u,v_0}^2+\frac{1}{p_{u_0}}(\sum_{v}I_{u_0,v})^2\nonumber\\
& \quad +\sum_{v:v\neq v_0}  \frac{p_{u_0v_0}}{p_{u_0v}p_{v_0}}I_{u_0,v}^2+\frac{1}{p_{v_0}}(\sum_{u}I_{u,v_0})^2.
\label{eqnI:M2}
\end{align}
One can verify that for our given choice of $I_{u,v}$ the right hand side of the equation \eqref{eqnI:M2} is equal to the left hand side of equation \eqref{eqnI:M1}, i.e.
\begin{align*}&\sum_{u,v:~u=u_0~or~v=v_0}\frac{1}{p_{uv}}I_{u,v}^2 = \frac{1}{p_{v_0}}(\sum_{u}I_{u,v_0})^2\\
& \qquad\qquad +\frac{1}{p_{u_0}}(\sum_{v}I_{u_0,v})^2 +\sum_{v:v\neq v_0}  \frac{p_{u_0v_0}}{p_{u_0v}p_{v_0}}I_{u_0,v}^2\\&\qquad\qquad +\sum_{u:u\neq u_0} \frac{p_{u_0v_0}}{p_{uv_0}p_{u_0}}I_{u,v_0}^2.
\end{align*}

This implies that both equations \eqref{eqnI:M2} and \eqref{eqnI:M1} have to hold with equality for our choice of $I_{u,v}$. Therefore all the inequalities that we took from Lemma \ref{le:kld} have to hold with equality. But this can happen {\em only if} $U$ is independent of $Y$, and $V$ is independent of $Z$, i.e. $I(U;Y)= I(V;Z)=0$. This is a contradiction and completes the proof.
\end{proof}

\begin{lemma} \label{le:kld2} If there are $u_1\neq u_2$ such that $f(u_1,v)=f(u_2,v)$ for all $v\in \mathcal{V}$, one can find another optimizer $p(u',v|x)$, where $I(U;Y)+I(V;Z)-I(U;V)=I(U';Y)+I(V;Z)-I(U';V)$ and furthermore $|\mathcal{U}'|<|\mathcal{U}|$. A similar condition holds if one can find $v_1\neq v_2$ such that $f(u,v_1)=f(u,v_2)$ for all $u\in \mathcal{U}$.
\end{lemma}
\begin{remark}
This lemma shows that to compute $T(X)$ one only needs to consider functions $f(U,V)$ where each row (fixed $U$) has a distinct mapping; similarly for columns.
\end{remark}
\begin{proof} Assume that $\mathcal{U}=\{u_1, u_2, ..., u_k\}$. Define $U'$ as a random variable taking values in $\{2,3,...,k\}$ as follows: $U'=i$ if $U=u_i$ for $i\geq 3$, and $U'=2$ if $U=u_1$ or $U=u_2$. Note that $H(X|U'V)=0$ since $f(u_1,v)=f(u_2,v)$ for all $v\in \mathcal{V}$. It suffices to prove that
\vspace*{-0.05in}$$I(U;Y)+I(V;Z)-I(U;V)\leq I(U';Y)+I(V;Z)-I(U';V).$$
This is equivalent to showing that $I(U;V|U')\geq I(U;Y|U')$. Since $H(X|U'V)=0$, we have
\vspace*{-0.05in}$$I(U;V|U')=I(U;VX|U')=I(U;VXY|U')\geq I(U;Y|U').$$ \
This completes the proof.\end{proof}

\begin{lemma}\label{le:kld} Take arbitrary $u_1,u_2,v,x$ such that $f(u_1,v)=x$ and
$f(u_2,v)=x$. Then any maximizing distribution must satisfy
\begin{align*}
&\sum_{y}\frac{p_{y|x}^2}{p_{u_2y}}\geq \frac{p_{u_1v}}{p_{u_2v}p_{u_1}},\\
\end{align*}
Equality implies that $p_{y|x}=p_{y|u_2}=p_{y|u_1}$ for all $y$.
\end{lemma}
\begin{proof} We start with the first derivative condition to write
\begin{align*}
\log \frac{p_{u_1v}}{p_{u_2v}} &\leq \sum_{y} p_{y|x} \log
\frac{p_{u_1y}}{p_{u_2y}} + \sum_{z} p_{z|x} \log \frac{p_{vz}}{p_{vz}}\\&=\sum_{y} p_{y|x} \log
\frac{p_{u_1y}}{p_{u_2y}}\\&=\sum_{y} p_{y|x} \log
\frac{p_{u_1}p_{y|u_1}}{p_{u_2y}}\\&=\sum_{y} p_{y|x} \log
\frac{p_{u_1}p_{y|x}}{p_{u_2y}}-\sum_{y} p_{y|x} \log
\frac{p_{y|x}}{p_{y|u_1}}\\&=\sum_{y} p_{y|x} \log
\frac{p_{u_1}p_{y|x}}{p_{u_2y}}-D(p_{y|x}\|p_{y|u_1})\\&\leq
\sum_{y} p_{y|x} \log
\frac{p_{u_1}p_{y|x}}{p_{u_2y}}\\&\leq
\log \sum_{y} p_{y|x}
\frac{p_{u_1}p_{y|x}}{p_{u_2y}}.
\end{align*}
\end{proof}

\subsection{An alternate proof for the XOR case}
\label{sse:xor}
In this section we provide an alternative proof for the binary XOR case, and its  generalization to the non-binary case (another extension of the XOR case has been provided in \cite{gea10}). Let us begin with the binary XOR case. Let $U, V$ be binary random variables, and $X=U\oplus V$. We would like to show that under this setting, we have
$$I(U;Y)+I(V;Z)\leq \max(I(X;Y), I(X;Z)).$$

\begin{definition} Given $p(u,x)$, let $c_{p(u,x)}$ denote the minimum value of $c$ such that
$I(U;Y)\leq c\cdot I(X;Y)$ holds for \emph{all} $p(y|x)$ for \emph{all} possible alphabets $\mathcal{Y}$. Alternatively, $c_{p(u,x)}$ is the minimum value of $c$ such that
the function $q(x) \mapsto H(U)-cH(X)$ when $p(u|x)$ is fixed, matches its convex envelope at $p(x)$.
\end{definition}
By the data-processing inequality we know that $0 \leq c_{p(u,x)} \leq 1$, and the minimum is well defined.
\begin{remark}
  Note that here we are adopting a dual notion. We fix the auxiliary channel $p(u|x)$ and then ask for a minimizing $c$ over all the forward channels.
\end{remark}
If $c_{p(u,x)}+c_{p(v,x)}\leq 1$ then note: $I(U;Y)+I(V;Z)\leq  c_{p(u,x)}I(X;Y)+c_{p(v,x)}I(X;Z) \leq \max(I(X;Y), I(X;Z))$.

\begin{theorem} \label{theoremEXOR} For any binary $U, V, X$ and $p(u,v,x)$ such that $X=U\oplus V$ the following inequality holds: $c_{p(u,x)}+c_{p(v,x)}\leq 1$.
\end{theorem}
\begin{proof}  Let $p_{ij}=p(U=i,V=j)$ for $i,j\in\{0,1\}$. Let $\alpha:=\frac{p_{00}}{p_{00}+p_{11}}=\frac{p_{00}}{p(X=0)}$ and $\beta:=\frac{p_{01}}{p_{01}+p_{10}}=\frac{p_{01}}{p(X=1)}$. Then we claim that $c_{p(u,x)}\leq |\alpha-\beta|$ and $c_{p(v,x)}\leq |\alpha+\beta-1|$. This will complete the proof since $\alpha, \beta \in [0,1]$ implies
\begin{align*}& |\alpha-\beta|+|\alpha+\beta-1| \leq  1.\end{align*}
 To show that $c_{p(u,x)}\leq |\alpha-\beta|$, it suffices to show that $q(x) \mapsto H(U)-|\alpha-\beta| H(X)$ is convex at all $q(x)$. The proof for $c_{p(v,x)}\leq |\alpha+\beta-1|$ is similar. Note that $ H(U)-|\alpha-\beta| H(X)=h(\alpha q(0)+\beta q(1))-|\alpha-\beta|h(q(0))$ where $h(\cdot)$ is the binary entropy function. Thus, we need to look at the function $x \mapsto h(\alpha x+\beta (1-x))-|\alpha-\beta|h(x)$ for $x\in [0,1]$. The second derivative is $$-\frac{(\alpha-\beta)^2}{(\alpha x+\beta(1-x))(1-(\alpha x+\beta(1-x)))}+\frac{|\alpha-\beta|}{x(1-x)}.$$ We need to verify that the above expression is non-negative, i.e.
\begin{align*}&(\alpha x+\beta(1-x))(1-(\alpha x+\beta(1-x)))\geq |\alpha-\beta|x(1-x).\end{align*}
This is true because
\begin{align*}&(\alpha x+\beta(1-x))(1-(\alpha x+\beta(1-x)))\\&
=(\alpha x+\beta(1-x))(x(1-\alpha)+(1-x)(1-\beta))\\& \geq  \alpha x(1-x)(1-\beta)+\beta(1-x)x(1-\alpha)
\\&
=x(1-x)[\alpha (1-\beta)+\beta(1-\alpha)]
 \\& \geq
x(1-x)|\alpha (1-\beta)-\beta(1-\alpha)|\\&
=x(1-x)|\alpha -\beta|.\end{align*}
\end{proof}
\begin{remark} Note that the definition of $c_{p(u,x)}$ requires the constraint $I(U;Y)\leq c\cdot I(X;Y)$ to hold for \emph{all} channels $p(y|x)$. If the subchannel $p(y|x)$ is known to be an erasure channel (i.e. $Y$ is equal to $X$ with some probability and erased otherwise), we can get even smaller values for $c$ (here $\frac{I(U;X)}{H(X)}$).
\end{remark}
In the Appendix C, we give a geometric interpretation to above, which yields insights for higher cardinality alphabets.

\section{conclusion}
We propose a pathway for verifying the optimality of Marton\rq{}s inner bound by trying to determine properties of the extremal distributions. We establish some necessary conditions, extending the work in the binary input case. We also add to the set of sufficient conditions. We present a few conjectures whose verifications have immediate consequences for the optimality of Marton's region.

\section*{Acknowledgments}

The work of Amin Gohari  was partially supported by the Iranian National Science Foundation Grant 89003743.

The work of Chandra Nair was partially supported by the following grants from the University Grants Committee of the Hong Kong Special Administrative Region, China: a) (Project No. AoE/E-02/08), b) GRF Project 415810. He also acknowledges the support from the Institute of Theoretical Computer Science and Communications (ITCSC) at the Chinese University of Hong Kong.

Venkat Anantharam gratefully acknowledges Research support from the ARO MURI grant W911NF-
08-1-0233, Tools for the Analysis and Design of Complex
Multi-Scale Networks, from the NSF grant CNS-
0910702, from the NSF Science \& Technology Center
grant CCF-0939370, Science of Information, from Marvell
Semiconductor Inc., and from the U.C. Discovery
program.

\bibliographystyle{IEEEtran}
%\bibliography{mybiblio}

\appendix

\subsection{Proof of an inequality for BSSC}

We consider the binary skew-symmetric broadcast channel with $p=\frac 12$ shown in Figure \ref{fig:bsscbc}. For this channel we prove that for all $\alpha \geq 1$ and for all $(U,V) \to X \to (Y,Z)$ we have
$$ \alpha I(U;Y) + I(V;Z) - I(U;V) \leq \max\{\alpha I(X;Y), I(X;Z)\}.$$

The proof of this claim again uses the perturbation method. Using the same arguments as in \cite{nwg10} it is easy to deduce that the AND case is the only non-trivial case. (The impossibility of an XOR mapping being a local maximum carries over
for all binary input broadcast channels in this setting.)

\begin{claim}
Any $p(u,v)>0$ such that $X=U \wedge V$ cannot  maximize (for $\a \geq 1$)
$$ \a I(U;Y) + I(V;Z) - I(U;V), $$
for the BSSC channel for a fixed $p(x)$.
\end{claim}
\begin{proof}
  Consider a perturbation of the form $p_\eps(u,v)=p(uv)(1 + \eps L(u,v))$,
where $\sum_{u,v}L(u,v) p(u,v|x) = 0$ for all $x$.
The first derivative conditions for a local maximum imply that
$$ \sum_{z} \qmf(z|0) \log \frac{p_{vz}(1z) p_{uv}(00)}{p_{vz}(0z) p_{uv}(01)} = 0,$$
$$ \sum_y \qmf(y|0) \log \frac{p_{uy}(1y)^\a p_u(0)^{(\a-1)} p_{uv}(00)}{p_{uy}(0y)^\a p_u(1)^{(\a-1)} p_{uv}(10)} = 0.$$

For BSSC we have $\qmf(y|0) = \frac 12, y \in \{0,1\}$ and $\qmf(y|1) = 1, y=1$ and vice-versa for $Z$.

Substituting this into first derivative conditions we obtain
$$ \frac{p_{vz}(10) p_{uv}(00)}{p_{vz}(00) p_{uv}(01)} = 1. $$
$$ \frac{p_{uy}(10)^\a p_u(0)^{(\a-1)} p_{uv}(00)}{p_{uy}(00)^\a p_u(1)^{(\a-1)} p_{uv}(10)} \frac{p_{uy}(11)^\a p_u(0)^{(\a-1)} p_{uv}(00)}{p_{uy}(01)^\a p_u(1)^{(\a-1)} p_{uv}(10)} = 1. $$

The first of the above conditions is equivalent to
$$\frac{p_{vz}(10) p_{uv}(00)}{p_{vz}(00) p_{uv}(01)} =  \frac{(p_{uv}(01) + \frac 12 p_{uv}(11)) p_{uv}(00)}{(p_{uv}(00) + p_{uv}(10)) p_{uv}(01)} = 1. $$
or that
$$ \frac 12 p_{uv}(11) p_{uv}(00) = p_{uv}(10) p_{uv}(01). $$

The second of the above conditions can be written as
\begin{align*}
1 &= \frac{p_{uy}(10)^\a p_u(0)^{2(\a-1)} p_{uv}(00)^2}{p_{uy}(00)^\a p_u(1)^{2(\a-1)} p_{uv}(10)^2} \frac{p_{uy}(11)^\a }{p_{uy}(01)^\a } \\
& = \frac{p_{uv}(10) ^\a (p_{uv}(00)+p_{uv}(01))^{2(\a-1)} p_{uv}(00)^2}{(p_{uv}(00)+p_{uv}(01))^\a (p_{uv}(10) + p_{uv}(11))^{2(\a-1)} p_{uv}(10)^2}\\&\quad\quad\times \frac{(p_{uv}(10) + 2p_{uv}(11))^\a }{(p_{uv}(00)+p_{uv}(01))^\a } \\
& = \frac{(1 + 2\frac{p_{uv}(11)}{p_{uv}(10)})^\a  }{(1+\frac{p_{uv}(01)}{p_{uv}(00)})^2 ( 1+ \frac{p_{uv}(11)}{p_{uv}(10)})^{2(\a-1)} }.
\end{align*}

Let $x = \frac{p_{uv}(01)}{p_{uv}(00)}$. Then from the first
condition we have
$\frac{p_{uv}(11)}{p_{uv}(10)} = 2x$. The second condition becomes
$$ 1 = \frac{(1 + 4x)^\a  }{(1+x)^2 ( 1+ 2x)^{2(\a-1)} }. $$

The second derivative conditions imply the following. Note that the expression we are dealing with is essentially $$(\a-1)H(U) + H(UV) - \a H(UY) - H(VZ)$$, since $H(Y)$ and $H(Z)$ are fixed.

Hence we would like to show that for all valid multiplicative perturbations, i.e. perturbations with 
$E(L|X) = 0$, as above, we have
\begin{align*}(\a - 1)\E(\E(L|U)^2) + \E(L^2) &- \a \E(\E(L|U,Y)^2)\\&~~~ - \E(\E(L|V,Z)^2) \geq 0. \end{align*}

Let $I_{kl} := p_{uv}{kl} L(k,l)$ for $k, l \in \{0,1\}$. Observe that 
$I_{00} = -I_{01} + I_{10}$ follows from the condition $E(L|X) = 0$. Computing the terms
\begin{align*}
\E(\E(L|U)^2) &= \frac{I_{10}^2}{p_{uv}(00) + p_{uv}(01)} + \frac{I_{10}^2}{p_{uv}(10) + p_{uv}(11)}, \\
\E(L^2) &= \frac{I_{00}^2}{p_{uv}(00)} + \frac{I_{01}^2}{p_{uv}(01)} + \frac{I_{10}^2}{p_{uv}(10)}, \\
\E(\E(L|U,Y)^2) &= \frac{I_{10}^2}{p_{uv}(00) + p_{uv}(01)} +  \frac{I_{10}^2}{2p_{uv}(10)} +\\&\qquad\qquad  \frac{I_{10}^2}{2(p_{uv}(10) + 2 p_{uv}(11))},\\
\E(\E(L|V,Z)^2) &=  \frac{I_{01}^2}{p_{uv}(00)+p_{uv}(10)} +  \frac{2I_{01}^2}{2p_{uv}(01) +  p_{uv}(11)}.
\end{align*}

Let $G$ be the negative of the Hessian. Using $I_{00}^2 = 
I_{01}^2 + 2 I_{01} I_{10} + I_{10}^2$, the quadratic form defined by $G$ can be
written as $G_{00} I_{01}^2 + (G_{01} + G_{10})I_{01} I_{10} + 
G_{11} I_{10}^2$, where 
\begin{align*}
G_{00} & = \frac{1}{p_{uv}(01)} + \frac{1}{p_{uv}(00)} - \frac{1}{p_{uv}(00)+p_{uv}(10)}\\&\quad - \frac{2}{2p_{uv}(01) +  p_{uv}(11)}, \\
G_{01} = G_{10} & = \frac{1}{p_{uv}(00)},\\
G_{11} & =  \frac{\a - 1}{p_{uv}(00) + p_{uv}(01)} + \frac{\a - 1}{p_{uv}(10) + p_{uv}(11)} \\&\quad+ \frac{1}{p_{uv}(10)} + \frac{1}{p_{uv}(00)} - \frac{\a}{p_{uv}(00) + p_{uv}(01)}\\
& \qquad - \frac{\a}{2p_{uv}(10)} - \frac{\a}{2(p_{uv}(10) + 2 p_{uv}(11))} \\
& = \frac{1}{p_{uv}(10)} + \frac{1}{p_{uv}(00)}  - \frac{1}{p_{uv}(00) + p_{uv}(01)} \\
& \qquad + \frac{\a - 1}{p_{uv}(10) + p_{uv}(11)}  - \frac{\a}{2p_{uv}(10)} \\&\quad- \frac{\a}{2(p_{uv}(10) + 2 p_{uv}(11))}.
\end{align*}

Using $p_{uv}(00) p_{uv}(11) = 2 p_{uv}(10) p_{uv}(01)$ we can write the term $G_{00}$ as
\begin{align*}
G_{00} & = \frac{1}{p_{uv}(01)} + \frac{1}{p_{uv}(00)} - \frac{1}{p_{uv}(00)+p_{uv}(10)} \\&\qquad- \frac{2}{2p_{uv}(01) +  p_{uv}(11)} \\
& = \frac{1}{p_{uv}(01)} + \frac{1}{p_{uv}(00)} - \frac{1}{p_{uv}(00)+p_{uv}(10)} \\&\qquad- \frac{p_{uv}(00)}{p_{uv}(01)(p_{uv}(00) +  p_{uv}(10))} \\
& = \frac{(p_{uv}(00) + p_{uv}(01)) p_{uv}(10)}{(p_{uv}(00) + p_{uv}(10)) p_{uv}(01) p_{uv}(00)}.
\end{align*}

Hence for $G$ to be positive semi-definite we require
\begin{align*} &\frac{\a - 1}{p_{uv}(10) + p_{uv}(11)}  - \frac{\a}{2p_{uv}(10)} - \frac{\a}{2(p_{uv}(10) + 2 p_{uv}(11))} \geq\\&\qquad - \frac{p_{uv}(00)}{p_{uv}(10)(p_{uv}(00) +  p_{uv}(01))}. \end{align*}

Multiplying by $p_{uv}(10)$ on both sides and noting that $x=\frac{p_{uv}(01)}{p_{uv}(00)} = \frac 12 \frac{p_{uv}(11)}{p_{uv}(10)}$ we can rewrite this necessary condition as
$$ \frac{\a - 1}{1 + 2x}  - \frac{\a}{2} - \frac{\a}{2(1 + 4x)} \geq - \frac{1}{(1 +  x)}.$$

This reduces to
$$ \alpha \leq \frac{1 + 4x}{(1+x)4x}. $$

Thus for a local maximum to occur we need that there is an $x \in (0, \infty)$ satisfying the following two conditions:
\begin{align*}
&1 \leq \a \leq \frac{1 + 4x}{(1+x)4x}, \\
 &1 = \frac{(1 + 4x)^\a  }{(1+x)^2 ( 1+ 2x)^{2(\a-1)} }.
\end{align*}

The second condition implies that
$$ \a \log \frac{(1 + 4x)  }{ ( 1+ 2x)^{2} } = \log \frac{(1+x)^2}{( 1+ 2x)^{2}}. $$

For $1  \leq \frac{1 + 4x}{(1+x)4x}$ (from first condition) to hold for some $x \in (0, \infty)$ we need $x \in (0, \frac 12]$.

Plugging the value of $\a$ from the second condition, we also require that
$$ \log \frac{(1+x)^2}{( 1+ 2x)^{2}} \geq \frac{1 + 4x}{(1+x)4x} \log \frac{(1 + 4x)  }{ ( 1+ 2x)^{2} }. $$
(Note the negativity of $\log \frac{(1 + 4x)  }{ ( 1+ 2x)^{2} }$ when $x \in (0, \frac 12]$.)

This is equivalent to
$$ \frac{(1+x)^2}{( 1+ 2x)^{2}} \geq \Big(\frac{(1 + 4x)  }{ ( 1+ 2x)^{2} }\Big)^{\frac{1 + 4x}{(1+x)4x}}. $$

Define
$$ g(x) =  \frac{(1+x)^2}{( 1+ 2x)^{2}} - \Big(\frac{(1 + 4x)  }{ ( 1+ 2x)^{2} }\Big)^{\frac{1 + 4x}{(1+x)4x}}. $$

Plotting $g(x)$ in the interval $[0, \frac 12]$ we see  that $g(0)=0$ and it is strictly negative and decreasing in $(0, \frac 12)$. Hence there is no $x$ simultaneously satisfying both the first and second derivative conditions for BSSC when $\a > 1$. This establishes the AND case for BSSC.

\begin{figure}
  \begin{center}
    \includegraphics[width=0.5\linewidth]{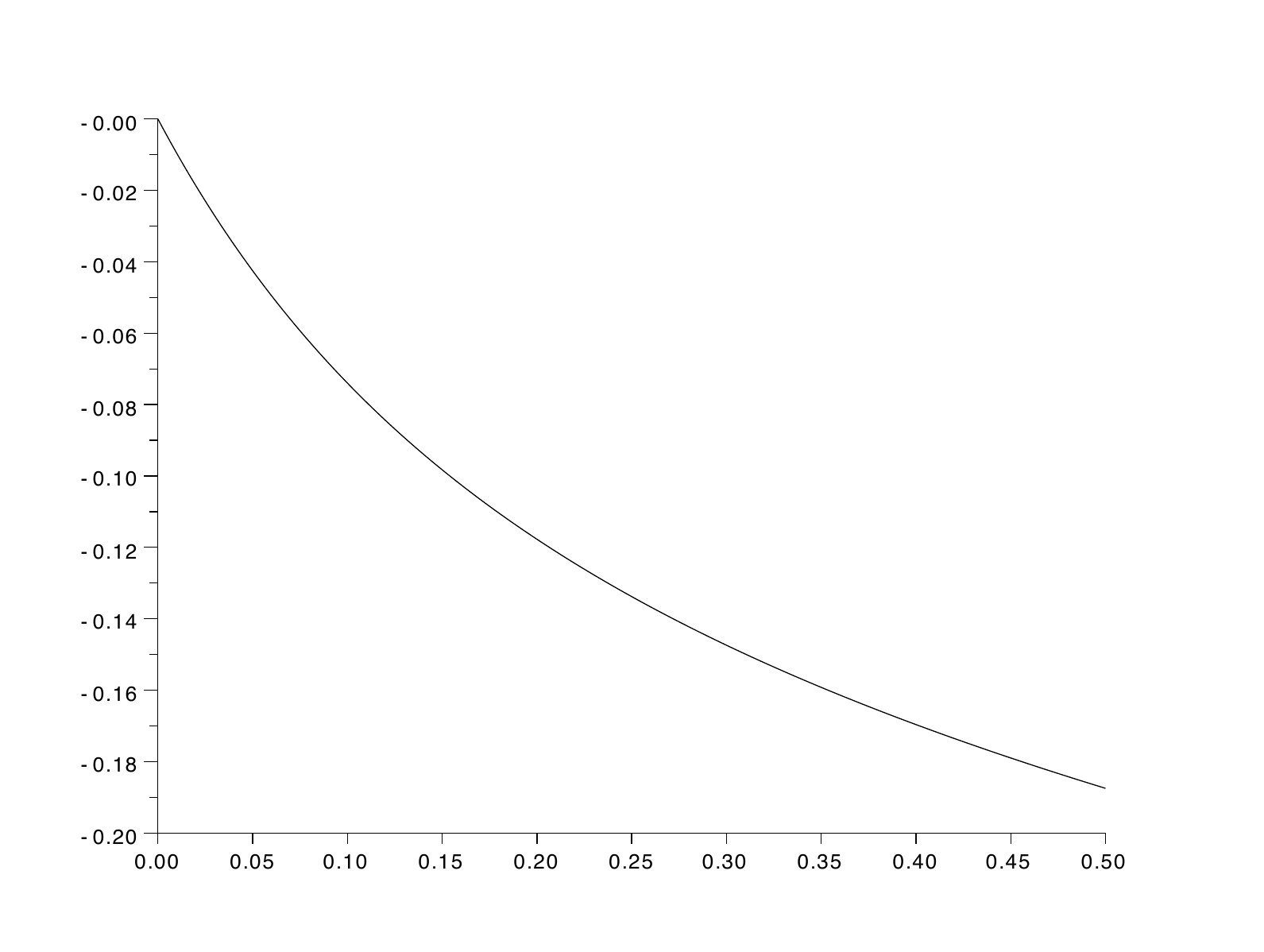}
  \end{center}
  \caption{The plot of $g(x)$ for $x \in [0, \frac 12]$}
\end{figure}

\end{proof}

\subsection{A counterexample}

We will produce a counterexample to the following statement: for $|\c X|=2$ the following inequality
$$\alpha I(U;Y) + I(V;Z) -I(U;V) \leq \max \{\alpha I(X;Y),I(X;Z)\},$$
holds for any $\alpha>1$ and any Markov chain $(U,V)\to X\to (Y,Z)$.

Consider the following setting: The channels are:
\begin{align*}
p(Y|X) =
\begin{bmatrix}
0.5 & 0.5 \\
0 & 1
\end{bmatrix}, \quad
p(Z|X) =
\begin{bmatrix}
1 & 0 \\
0.1 & 0.9
\end{bmatrix}.
\end{align*}
The parameters are
$$p(X) = [0.8, ~ 0.2], \quad \alpha = \frac{I(X;Z)}{I(X;Y)} = 3.429517.$$
The choice of $\alpha$ is actually the corner point for $RHS$.
The result is
$$LHS = 0.593020 > 0.586278 = RHS,$$
where $LHS$ is obtained by the p.m.f. and mapping $X=f(U,V)$ as
\begin{align*}
p(U,V) =
\begin{bmatrix}
0.05930 & 0.00005 \\
0.14065 & 0.80000
\end{bmatrix}, \quad
f(U,V) =
\begin{bmatrix}
1 & 1 \\
1 & 0
\end{bmatrix}.
\end{align*}
This is AND case, exactly the case we cannot prove in general.

We also plot the figure for $LHS$ and $RHS$ w.r.t. $\alpha$:
\begin{figure}[h]
	\centering
	\includegraphics[width=0.5\textwidth]{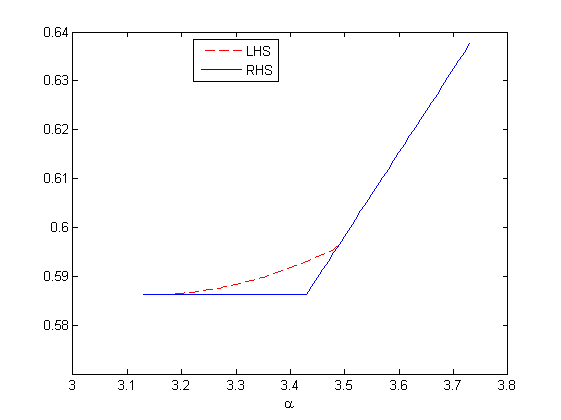}
\end{figure}

{\em Acknowledgment}: The authors wish to thank Yanlin Geng for obtaining this counterexample.

\subsection{A geometric interpretation for  $c_{p(u,x)}$}

We know that $c_{p(u,x)}$ is the minimum value of $c$ such that
the function $q(x) \mapsto H(U)-cH(X)$ when $p(u|x)$ is fixed, matches its convex envelope at $p(x)$. The $\max_{p(x)}c_{p(u,x)}$ would be the minimum value of $c$ such that $q(x) \mapsto H(U)-cH(X)$ is completely convex in $p(x)$ for a fixed $p(u|x)$. Let us fix some $p(u|x)$ and let $c'_{p(u,x)}$ be the minimum value of $c$ such that
the function $q(x) \mapsto H(U)-cH(X)$ is convex at $p(x)$. We will then have that $\max_{p(x)}c_{p(u,x)}=\max_{p(x)}c'_{p(u,x)}$.

The term $c'_{p(u,x)}$ has a nice geometric interpretation. We begin by using the perturbation method and perturb $p(u,x)$ along a direction $L(X)$ such that $E[L]=0$ (therefore we are fixing $p(u|x)$). The second derivative of $H(U)-c H(X)$ along this direction is
$$-E(E(L|U)^2)+c E(E(L|X)^2) = -E(E(L|U)^2)+c E(L^2).$$
If we want this to be greater than or equal to zero, it implies that $c \geq \frac{E(E(L|U)^2)}{E(L^2)}$ for all $L(X)$ such that $E[L]=0$.

%Now consider the vector space of random variable, $\mathcal{V}$ on a given sample space $\Omega$ defined as follows: 
Given a fixed sample space, $\Omega$, the set of all square-integrable real-valued random variables 
%(i.e. all functions from the sample space to real numbers) 
forms a vector space $\mathcal{V}$ with normal addition and scalar multiplication of random variables. We define the inner product between two random variables X and Y to be $\mathbb{E}(XY)$.

The set of all real-valued functions of $X$
%, i.e. $\{L(X): L:\mathbb{R}\mapsto \mathbb{R}\}$
is itself a linear subspace of $\mathcal{V}$. Let us denote this set 
%of all random variables that are functions of $X$ 
by $\mathcal{V}_X$. We can similarly define the set of all real-valued 
random variables that are functions of $U$ and denote it by $\mathcal{V}_U$.

Now, let us use $\textbf{1}$ to denote the random variable that takes value 1 with probability 1. The set of 
square-integrable real-valued random variables that are perpendicular to $\mathbf{1}$ are the ones with zero expected value. Let us denote the set of these random variables, itself a linear subspace, by $\mathcal{V}_{\textbf{1}\bot}$.

Note that $\textbf{1}\in\mathcal{V}_X$ and $\textbf{1}\in\mathcal{V}_U$. Let us define the following two subspaces:
$$\mathcal{V}'_X=\mathcal{V}_X \cap \mathcal{V}_{\textbf{1}\bot},$$
$$\mathcal{V}'_U=\mathcal{V}_U \cap \mathcal{V}_{\textbf{1}\bot}.$$
Now, the perturbation method says that we should take some $L(X)$ in $\mathcal{V}'_X$. Its projection onto $\mathcal{V}_U$ is equal to $E(L|U)$, its squared length being $E(E(L|U)^2)$. Note that the projection onto $\mathcal{V}_U$ is the same as the projection onto $\mathcal{V}'_U$ because all the action is taking place in $ \mathcal{V}_{\textbf{1}\bot}$. The expression $\frac{E(E(L|U)^2)}{E(L^2)}$ is the cosine squared of the angle formed by vector $L$ and its projection onto $\mathcal{V}'_U$. The term $c $ should dominate all such cosine-squared values when $L$ freely changes over $\mathcal{V}'_X$. Thus, $c'_{p(u,x)}$ has to be the cosine-squared of the angle between the two subspaces $\mathcal{V}'_U$ and $\mathcal{V}'_X$. This is because we are taking an arbitrary vector $L$ in $\mathcal{V}'_X$, then finding the vector in $\mathcal{V}'_U$ that has the smallest angle with $L$, i.e. its projection of $L$ onto $\mathcal{V}'_U$, and then computing their cosine-squared expression.

Note that if the Gacs-Korner common information between $U$ and $X$ is non-trivial, then the angle between the two subspaces $\mathcal{V}'_U$ and $\mathcal{V}'_X$ is zero (because the intersection of $\mathcal{V}'_X$ and $\mathcal{V}'_U$ will be non-trivial). Otherwise, the angle between the two subspaces is strictly positive. It is worth noting that the angle between the two subspaces $\mathcal{V}'_U$ and $\mathcal{V}'_X$ has a \emph{symmetric} definition. Therefore the minimum value of $c$ such that the function $q(x)\mapsto H(U)-cH(X)$ is convex at the given $p(x)$, is the same as the minimum value of $c$ such that $q(u)\mapsto H(X)-cH(U)$ is convex at the given $p(u)$.

To compute the cosine of the angle between the two subspaces it suffices to take two arbitrary vectors in these two subspaces and maximize the cosine of their angle. We can express the cosine of the angle using Pearson's correlation coefficient between two variables:
$$c'_{q(u,x)}=(\max_{L:\mathcal{X}\mapsto \mathbb{R},~T:\mathcal{U}\mapsto \mathbb{R}}\frac{Cov\big(L(X), T(U)\big)}{\sqrt{Var(L(X))}\times \sqrt{Var(T(U))}})^2.$$
Note that here the maximization is over arbitrary functions $L$ and $T$, and the requirement that $\mathbb{E}[L]=\mathbb{E}[T]=0$ is relaxed. Another formula for $c'_{q(u,x)}$ is the maximum of $(\mathbb{E}[L(X)T(U)])^2$ over all $L(X)$ and $T(U)$ satisfying $\mathbb{E}[L(X)]=0$, $\mathbb{E}[L(X)^2]=1$, $\mathbb{E}[T(U)]=0$ and $\mathbb{E}[T(U)^2]=1$. This is a simple optimization problem and can be dealt with using the Lagrange multipliers technique. It gives rise to other analytical formulas for $c'_{q(u,x)}$.

\end{document}